\documentclass[article]{siamart0516}

\usepackage{cite}
\usepackage{lipsum}
\usepackage{amsfonts}
\usepackage{graphicx}
\usepackage{epstopdf}
\usepackage{algorithmic}
\usepackage{amsopn}
\usepackage[table]{colortbl}
\ifpdf
  \DeclareGraphicsExtensions{.eps,.pdf,.png,.jpg}
\else
  \DeclareGraphicsExtensions{.eps}
\fi

\newtheorem{example}{Example}
\newtheorem{conjecture}{Conjecture}

\newcommand{\F}{\mathbb{F}}

\DeclareMathOperator{\diag}{diag}

\DeclareMathOperator{\supp}{supp}

\DeclareMathOperator{\Rep}{Rep}
\DeclareMathOperator{\lcm}{lcm}

\title{{Private Information Retrieval from Coded Databases with Colluding Servers}\thanks{Submitted to the editors \today.
\funding{D.\ Karpuk is supported by Academy of Finland grant 268364. C.\ Hollanti is supported by Academy of Finland grants 276031, 282938, and 303819. \newline 
All authors are with the Department of Mathematics and Systems Analysis, Aalto University, Finland.}}}

\author{
  Ragnar Freij-Hollanti\thanks{Email: \email{ragnar.freij@aalto.fi}}
  \and Oliver W.~Gnilke\thanks{Email: \email{oliver.gnilke@aalto.fi}}
  \and Camilla Hollanti\thanks{Email: \email{camilla.hollanti@aalto.fi}}
  \and David A.~Karpuk\thanks{Email: \email{david.karpuk@aalto.fi}}
}

\begin{document}

\maketitle

\begin{abstract}
We present a general framework for Private Information Retrieval (PIR) from arbitrary coded databases, that allows one to adjust the rate of the scheme to the suspected number of colluding servers.   If the storage code is a generalized Reed-Solomon code of length $n$ and dimension $k$, we design PIR schemes that achieve a PIR rate of $\frac{n-(k+t-1)}{n}$ while protecting against any $t$ colluding servers, for any $1\leq t\leq n-k$. This interpolates between the previously studied cases of $t=1$ and $k=1$ and achieves PIR capacity in both of these cases asymptotically as the number of files in the database grows. 
\end{abstract}



\section{Introduction}
Private information retrieval (PIR) addresses the question of how to
retrieve data items from a database without disclosing information
about the identity of the data items retrieved, and was introduced by Chor, Goldreich,
Kushilevitz and Sudan in
\cite{chor1995private, chor1998private}. The classic PIR model of \cite{chor1998private}
views the database as an $m$-bit binary string $x = [x^1\cdots x^m]
\in \{0,1\}^m$, and assumes that the user wants to retrieve a single bit
$x^i$ without revealing any information about the index $i$.  We consider a natural extension of this model, wherein the database is a string $x = [x^1 \cdots x^m]$ of files $x^i$, which are themselves bit strings, and the user wants to download one of the files $x^i$, without revealing its index. 

The rate of a PIR scheme in this model is measured as the ratio of the gained information over the downloaded information, while upload costs of the requests are usually ignored.
 The trivial solution is to download the entire database. This, however, incurs a
significant communication overhead whenever the database is large, and
is therefore not useful in practice. 
While the trivial solution is the only way to guarantee \emph{information-theoretic
privacy} in the case of a single server~\cite{chor1998private}, this problem can be remedied by replicating the database onto $k$
servers that do not communicate.

The study of PIR recently received renewed attention, when Shah, Rashmi, Ramchandran and Kumar introduced a model of coded private information retrieval (cPIR)~\cite{shah2012, shah2014}. Here, all files are distributed over the servers according to a storage code, so there is no assumption that the contents of all servers are identical. It is shown in~\cite{shah2014} that for a suitably constructed storage code, privacy can be guaranteed by downloading a single bit more than the size of the desired file. However, this requires exponentially many servers in terms of the number of files. Blackburn, Etzion and Paterson achieved the same low download complexity with a linear number of servers~\cite{blackburn2016small}. This is a vast improvement upon~\cite{shah2014}, but still far from applicable storage systems where the number of files tends to dwarf the number of servers.

While~\cite{shah2014} effectively answered the question on how low the communication cost of a PIR scheme can be, it highlighted another cost
parameter that should not be neglected in the era of big data, namely the \emph{storage
overhead}. We define the storage overhead as the ratio of the total number of coded bits
stored on all servers to the total number of uncoded bits of data. Fazeli, Vardy and Yaakobi showed in \cite{fazeli2015pir} that it is possible to reduce the
storage overhead significantly.
However, this requires subpacketizing the file and distributing it over a number of servers that grows to infinity as the desired storage overhead decreases. 

In contrast to the schemes in~\cite{shah2014}, whose strengths appear as the number of servers tends to infinity,  we are considering the setting where we are given a storage system with a fixed number of servers. While this in no way optimizes the storage overhead, it does keep the overhead fixed. Moreover, we allow some subsets of servers to be \emph{colluding}, by which we mean that they may inform each other of their interaction with the user. This is very natural in a distributed storage system where communication between servers is required to recover data in the case of node failures.
PIR over fixed maximum distance separable (MDS) storage systems were considered in~\cite{razan_salim}. There, two PIR schemes were presented for arbitrary $[n,k]$ MDS codes, one of which had rate $\frac{1}{n}$ and protected against $t=n-k$ colluding servers, and the other of which had rate $\frac{n-k}{n}$ and $t=1$. In Section \ref{main_scheme} we present these as special instances of a scheme that can handle any number $1\leq t\leq n-k$ of colluding servers. Curiously, neither the performance of our scheme nor the underlying field size depends on the number of files stored.
The rate of our scheme depends on the minimum distance of a certain star product, and in the case where the storage code is a generalized Reed Solomon (GRS) code (Theorem~\ref{theorem2}), we can achieve a rate of $\frac{n-(k+t-1)}{n}$. 

 The capacity (i.e.\ maximum possible rate) of a PIR scheme for a replicated storage system was derived in~\cite{sun_jafar_1} (without collusion) and~\cite{sun_jafar_2} (with colluding servers). The corresponding capacity of a coded storage system was given in~\cite{bananaman}, in the case of no colluding servers. A previous version of this article conjectured a formula for the capacity of coded PIR with colluding servers, that gives the capacity bounds of~\cite{sun_jafar_2, bananaman} as special cases.  However, this conjecture was recently disproven by Sun and Jafar in \cite{sun_jafar_MDS_TPIR}, who gave an explicit example of a scheme for coded PIR with colluding servers the rate of which exceeded our conjectured upper bound.  The work of \cite{sun_jafar_MDS_TPIR} also characterizes the capacity of PIR for MDS coded data with colluding servers for other parameters, in particular for $m = 2$ files of length $k = n-1$. In general, however, the capacity of coded PIR with colluding servers remains open.  Recently, another PIR scheme for coded databases with colluding servers was presented in \cite{zhang_ge} for MDS coded data.  However, the rates achieved in the present work outperform the rates presented in \cite{zhang_ge} even for a moderate number of files, and the scheme in \cite{zhang_ge} requires the underlying field to be large. The case of a linear storage code that is not necessarily MDS, without collusion, has been studied in \cite{Kumar}. They show that for certain codes a rate of $\frac{n-k}{n}$ can still be achieved, outperforming the general result in Theorem~\ref{correct}.

\section{Coding-Theoretic Preliminaries}

In this section we briefly collect some standard coding-theoretic definitions and results that we will need in the following sections. 

\subsection{Basic Definitions}

We will use $\mathbb{F}_q$ to denote the field with $q$ elements, where $q$ is a prime power, and $[n]$ for the set $\{1,2,\ldots , n\}$.  For any two vectors $v,w\in \F_q^n$, we denote their standard inner product by $\langle v,w \rangle$.  If $V\subseteq \F_q^n$, then we denote its orthogonal complement by 
\begin{equation}
V^\perp = \{w\in \F_q^n\ |\ \langle v,w\rangle = 0\ \text{for all $v\in V$}\}
\end{equation}
and we write $V\perp W$ if $W\subseteq V^\perp$.

For a code $C\subseteq \mathbb{F}^n_q$ or a vector $v\in \F_q^n$ we use $C_I$ and $v_I$ to denote their respective projections onto the coordinates in $I\subseteq[n]$. The support of a codeword $c\in \F^n_q$ is $\supp(c)=\{i\in [n]: c_i\neq 0\}$, and the support of a code $C\subseteq \mathbb{F}^n_q$ is $\supp(C)=\cup_{c\in C}\supp(c)$.  The minimum distance of $C\subseteq \mathbb{F}^n_q$ is 
\begin{equation}
d=d_C = \min\{|I| : |C_{[n]\setminus I}|<|C|\}.
\end{equation}
For linear codes, this can alternatively be written as
\begin{equation}
d =\min\{|\supp(c)| : c\in C\}.
\end{equation}

A linear code $C\subseteq \mathbb{F}^n_q$ of dimension $k$ and with minimum distance $d$ is called an $[n,k,d]$-code, or an $[n,k,d]_q$-code if we wish to emphasize the field of definition. By an elementary result that is usually attributed to Singleton~\cite{Singleton}, if $C$ is an $[n,k,d]$-code, then 
\begin{equation}\label{singleton} 
d\leq n - k + 1
\end{equation} 
A code that satisfies~\eqref{singleton} with equality is called a \emph{maximum distance separable} (MDS) code.  An $[n,k,d]$ MDS code will be more concisely denoted as an $[n,k]$ MDS code, with $d = n-k+1$ being implied.

Given a linear $[n,k,d]$-code $C$, a subset $K\subseteq[n]$ of size $|K| = k$ is an \emph{information set} of $C$ if the natural projection $C\rightarrow C_K$ is a bijection.  Equivalently, the columns of any generator matrix of $C$ corresponding to the indices in $K$ are linearly independent.  If $C$ is an MDS code, then every $K\subseteq [n]$ of size $k$ is an information set.

The repetition code $\Rep(n)_q\subseteq\F_q^n$ is the one-dimensional code generated by the all-ones vector. It is an $[n,1]$ MDS code.

\subsection{Generalized Reed-Solomon Codes}  Our proposed PIR scheme will be most interesting when the code defining the storage system is a Generalized Reed-Solomon code.  As such, we recall the basic properties of such codes here.

\begin{definition}[GRS Codes]
	Let $\alpha=[\alpha_1\cdots \alpha_n] \in \mathbb{F}^n_q$ satisfy $\alpha_i \neq \alpha_j$ for $i \neq j$, and let $v=[v_1\cdots v_n] \in {\mathbb{F}_q^\times}^n$. We define the Generalized Reed-Solomon (GRS) code of dimension $k$ associated to these $n$-tuples to be
	\begin{equation}GRS_k(\alpha,v)=\left\{ (v_i f(\alpha_i))_{1 \leq i \leq n}\ |\ f\in \F_q[x],\ \deg(f)<k \right\}.\end{equation}
	\end{definition}
	
	The canonical generator matrix for this code is given by
	\begin{equation}\label{eq:canonical} G(\alpha,v):=\begin{bmatrix}
		1 & \cdots & 1\\
		\alpha_1 & \cdots & \alpha_n \\
		\alpha_1^2 & \cdots & \alpha_n^2 \\
		
		\vdots & \ddots & \vdots \\
		\alpha_1^{k-1} & \cdots & \alpha_n^{k-1} \\
	\end{bmatrix}
	\cdot \diag(v),
	 \end{equation}
	 where $\diag(v)$ is the diagonal matrix with the values $v_i$ on the diagonal. In data storage applications, it is often desirable to have an explicit encoding matrix that is systematic, \emph{i.e.} having an identity submatrix in the first $k$ columns. For this purpose, define
	 \begin{equation}\label{eq:systematic} \tilde{G}(\alpha,v):=\begin{bmatrix}
		f_1(\alpha_1) & \cdots & f_1(\alpha_n)\\
		\vdots & \ddots & \vdots \\
		f_k(\alpha_1) & \cdots & f_k(\alpha_n) \\
	\end{bmatrix}
	\cdot \diag(v),
	 \end{equation}
where \[f_i(x)=  v_i^{-1}\prod_{j\in [k]\setminus\{i\}}\frac{x-\alpha_j}{\alpha_i-\alpha_j}\] for $i=1,\ldots , k$. Note that $f_i(\alpha_i)= v_i^{-1}$ and $f_i(\alpha_j)=0$ for $j\in[k]\setminus\{i\}$, hence this is a systematic generator matrix for $GRS_k(\alpha,v)$.

	The code $GRS_k(\alpha,v)$ is an $[n,k]$ MDS code. From the Lagrange interpolation formula, it follows that the dual of $GRS_k(\alpha, v)$ is given by $GRS_{n-k}(\alpha,u)$ where \begin{equation}
	u_i=(v_i \prod_{j \neq i}(\alpha_i-\alpha_j))^{-1}\label{inverse_coeff}.
	\end{equation}

\subsection{Star Products}  The star product of two codes will play an integral role in  our PIR scheme, essentially determining its rate.

\begin{definition}
	Let $V, W$ be sub-vector spaces of $\F_q^n$. We define the star (or Schur) product $V \star W$ to be the subspace of $\F_q^n$ generated by the Hadamard products $v \star w = [v_1w_1 \cdots v_nw_n]$ for all pairs $v \in V, w \in W$.
\end{definition}

The following proposition collects some basic properties of the star product that will prove useful in the coming sections.

\begin{proposition}\label{star}
The star product satisfies the following properties:
\begin{itemize}
\item[(i)] If $C$ is any linear code in $\F_q^n$ and $\Rep(n)_q\subseteq \F_q^n$ is the repetition code of length $n$ over $\F_q$, then $C\star \Rep(n)_q = C$.
\item[(ii)] If $C$ and $D$ are any linear codes in $\F_q^n$ with $\supp(C)=\supp(D)=[n]$, and $(C\star D)^{\perp} = H$, then $d_H\geq d_{C^\perp}+d_{D^\perp}-2$.
\item[(iii)] If $C\subseteq \F_q^n$ is any MDS code, then $(C\star C^\perp)^\perp = \Rep(n)_q$.
\item[(iv)] The star product of two generalized Reed-Solomon codes in $\F_q^n$ with the same parameter $\alpha$ is again a generalized Reed-Solomon code with parameter $\alpha$.  More specifically, $GRS_k(\alpha,v)\star GRS_\ell(\alpha,w) = GRS_{\min\{k+\ell-1,n\}}(\alpha,v\star w)$ for any parameters $v,w$.
\end{itemize}
\end{proposition}
\begin{proof}
Property (i) follows immediately from the definition of the star product.  Property (ii) is Theorem 5 of \cite{VanLint}. To see that property (iii) holds, let $H = (C\star C^\perp)^\perp$.  The containment $\Rep(n)_q\subseteq H$ is obvious for any code $C$.  If $C$ is an $[n,k]$ MDS code then $C^\perp$ is an $[n,n-k]$ MDS code, so property (ii) implies that $d_H \geq d_C+ d_{C^\perp} -2 = n$.   Hence by the Singleton bound the dimension of $H$ is $1$ and therefore $H = \Rep(n)_q$.  

To see that property (iv) holds, consider some arbitrary codewords $(v_if(\alpha_i))\in GRS_k(\alpha,v)$ and $(w_ig(\alpha_i))\in GRS_\ell (\alpha,w)$.  We clearly have
\begin{equation}
(v_if(\alpha_i))\star(w_ig(\alpha_i)) = (v_iw_i(fg)(\alpha_i)) \in GRS_{\min\{k+\ell-1,n\}}(\alpha,v\star w)
\end{equation}
hence the containment $GRS_k(\alpha,v)\star GRS_\ell(\alpha,w) \subseteq GRS_{\min\{k+\ell-1,n\}}(\alpha,v\star w)$ holds.  To see the reverse containment, note that $GRS_{\min\{k+\ell-1,n\}}(\alpha,v\star w)$ is generated as an $\F_q$-vector space by codewords of the form $(v_iw_if_m(\alpha_i))$ where $f_m(x) = x^m$ is a monomial of degree $m<k+\ell-1$.  We can clearly decompose such a codeword as
\begin{equation}
(v_iw_if_m(\alpha_i)) = (v_if_a(\alpha_i))\star(w_if_b(\alpha_i))
\end{equation}
where $f_a(x) = x^a$ and $f_b(x) = x^b$ for any $a,b$ such that $a<k$, $b<\ell$, and $a+b = m$.  This shows the reverse inclusion and completes the proof.
\end{proof}

\section{Coded Storage and Private Information Retrieval}\label{DSS}

Let us describe the distributed storage systems we consider; this setup follows that of \cite{razan_salim,bananaman}.  To provide clear and concise notation, we have consistently used superscripts to refer to files, subscripts to refer to servers, and parenthetical indices for entries of a vector.  So, for example, the query $q^i_j$ is sent to the $j^{th}$ server when downloading the $i^{th}$ file, and $y^i_j(a)$ is the $a^{th}$ entry of the vector $y^i_j$ stored on server $j$.

Suppose we have files $x^1,\ldots,x^m\in \F_q^{b\times k}$.  Data storage proceeds by arranging the files into a $bm \times k$ matrix
\begin{equation} \
X=
\begin{bmatrix}
x^1 \\ \vdots \\ x^m
\end{bmatrix}
\end{equation}
Each file $x^i$ is encoded using a linear $[n,k,d]$-code $C$ with generator matrix $G_C$ into an encoded file $y^i = x^iG_C$.  In matrix form, we encode the matrix $X$ into a matrix $Y$ by right-multiplying by $G_C$:
\begin{equation} Y= X G_C = 
\begin{bmatrix}
y^1 \\ \vdots \\ y^m
\end{bmatrix}
= 
\begin{bmatrix}
y_1 & \cdots & y_n
\end{bmatrix}
=
\begin{bmatrix}
y^1_1 & \cdots & y^1_n \\
\vdots & \ddots & \vdots \\
y^m_1 & \cdots & y^m_n
\end{bmatrix}
\end{equation}
The $j^{th}$ column $y_j\in\F_q^{bm\times 1}$ of the matrix $Y$ is stored by the $j^{th}$ server.  Here the vector $y^i_j\in \F_q^{b\times 1}$ represents the part of the $i^{th}$ file stored on the $j^{th}$ server.

Such a storage system allows any $d_C-1$ servers to fail while still allowing users to successfully access any of the files $x^i$.  In particular, if $C$ is an MDS code, the resulting distributed storage system is maximally robust against server failures.

\emph{Private Information Retrieval} (PIR) is the process of downloading a file from a database without revealing to the database which file is being downloaded \cite{chor1998private}.  Here by `database' we mean a collection of servers, e.g.\ all of the $n$ servers used in the distributed storage system described earlier.

\begin{definition}\label{PIR_def}
Suppose we have a distributed storage system as described above.  A \emph{linear PIR scheme over $\F_q$} for such a storage system consists of:
\begin{itemize}
\item[1.] For each index $i\in[m]$, a probability space $(\mathcal{Q}^i,\mu^i)$ of \emph{queries}.  When the user wishes to download $x^i\in \F_q^{b\times k}$, a query $q^i \in \mathcal{Q}^i$ is selected randomly according to the probability measure $\mu_i$.  Each $q^{i}$ is a set $q^{i} = \{q^{i}_1,\ldots,q^{i}_n\}$, where $q^{i}_j$ is sent to the $j^{th}$ server, and furthermore $q^i_j$ is itself a row vector of the form
\begin{equation}
q^i_j = [q^{i1}_j \cdots q^{im}_j]\quad \text{where}\quad q^{i\ell}_j\in \F_q^{1\times b},\ \text{for all $\ell \in [m]$.}
\end{equation}
\item[2.] \emph{Responses} $r^{i}_j = \langle q^{i}_j,y_j\rangle\in \F_q$ which the servers compute and transmit to the user.  We set $r^i = [r^i_1\cdots r^i_n]$ to be the total response vector.
\item[3.] An \emph{iteration process}, which repeats Steps 1.-2.\ a total of $s$ times until the desired file $x^i$ can be reconstructed from the $s$ responses $r^i$
\item[3.] A \emph{reconstruction function} which takes as input the various $r^i$ over all of the $s$ iterations and returns the file $x^i$.

\end{itemize}
\end{definition}

Here we view $b$ and $s$ as secondary parameters, which we are free to adjust to enable the user to download exactly one whole file.  If one restricts to the case of $b =  1$ row per file, the size of the file $k$ may be too small, in which case it is not clear how to take advantage of a high rate scheme which inherently downloads more symbols per iteration than there are in a file.  On the other hand, restricting to schemes with $s = 1$ iteration may fail to download an entire file.  The freedom to adjust $b$ and $s$ allows one to avoid such complications.


The \emph{rate} of a PIR scheme measures its efficiency, by comparing the size of a file with how much information we downloaded in total:

\begin{definition}\label{rate}
The \emph{rate} of a linear PIR scheme is defined to be $\frac{bk}{ns}$.
\end{definition}

Note that Definition \ref{rate} ignores the cost to the user of uploading the queries to the servers. This can be justified by considering $x^i\in V^{b\times k}$ where $V$ is some finite-dimensional vector space over $\F_q$, and encoding and data retrieval proceeds in an obvious way.  In this setting the size of the queries is easily seen to be minimal in comparison to download costs when $\dim V\gg 1$.

It would be more precise to define rate as $\frac{bk}{E[w(q^i)]s}$, where $E(\cdot)$ denotes expectation and $w(q^i)$ is the number of queries $q^i_j$ which are not the zero vector, since such queries can be ignored. For comparison with earlier work on PIR we use \cref{rate} for the remainder of this paper.


\begin{definition}\label{collusion}
A PIR scheme \emph{protects against $t$ colluding servers} if for every set $T=\{j_1,\ldots,j_t\}\subseteq[n]$ of size $t$, we have
\begin{equation}\label{mi}
I(Q^i_T;i) = 0
\end{equation}
where $Q^i_T$ denotes the joint distribution of all tuples $\{q^i_{j_1},\ldots,q^i_{j_t}\}$ of queries sent to the servers in $T$ as we range over all $s$ iterations of the PIR scheme, and $I(\cdot\ ;\cdot)$ denotes the mutual information of two random variables.  
\end{definition}

In other words, for every set $T$ of servers of size $t$, there exists a probability distribution $(\mathcal{Q}_T,\mu_T)$ such that for all $i\in [m]$, the projection of $(\mathcal{Q}^i,\mu^i)$ to the coordinates in $T$ is $(\mathcal{Q}_T,\mu_T)$. Hence, no subset of servers of size $t$ will learn anything about the index $i$ of the file that is being requested.  If a PIR scheme protects against $t$ colluding servers, it also clearly protects against $t'$ colluding servers for all $t'\leq t$.



\section{A General PIR Scheme for Coded Storage with Colluding Servers}\label{main_scheme}

Our goal is to find high-rate PIR schemes which protect against many colluding servers.  To that end, the following construction provides a general PIR scheme for coded databases which protects against a flexible number of colluding servers. 


\subsection{Scheme Construction}\label{construction} Let $C$ be a linear $[n,k,d]_q$ code with generator matrix $G_C$, and consider the distributed storage system $Y = XG_C$ as in Section \ref{DSS}.  We choose another linear code $D\subseteq \F_q^n$, the \emph{retrieval code}.  As we will see, the retrieval code essentially determines the privacy properties of the scheme.

Throughout this section, $i$ will denote the index of the file we wish to retrieve.  We begin by simplifying notation, defining
\begin{equation}
c:= d_{C\star D}-1.
\end{equation}
The queries are constructed so that the total response vector during one iteration is of the form
\begin{equation}
r^i = \text{(codeword of $C\star D$)} + \tilde{y}^i
\end{equation}
where $\tilde{y}^i$ is a vector containing $c$ distinct symbols of $y^i$ in known locations, and zeros elsewhere.  Multiplying $r^i$ by a generator matrix of $(C\star D)^\perp$ then allows us to recover these $c$ symbols.  

To allow the user to download exactly one file over $s$ iterations, we force the file size $bk$ to be an integer multiple of $c$ by setting
\begin{equation}
b = \frac{\lcm(c,k)}{k}\quad\text{and}\quad s = \frac{\lcm(c,k)}{c}
\end{equation}
so that $bk = sc$.  During each iteration of the scheme, we download
\begin{equation}
g : = \frac{k}{s} = \frac{c}{b}
\end{equation} 
symbols from every row of $y^i$.  After $s$ iterations, the scheme will have downloaded $sg = k$ symbols of the $a^{th}$ row $y^{i,a}$ of $y^i$ for all $a \in [b]$.  

We also fix a subset $J\subseteq [n]$ of servers of size
\begin{equation}
|J| = \max\{c,k\}
\end{equation}
which stays constant throughout the scheme.  By reindexing the servers if necessary, we may assume without loss of generality that
\begin{equation}
J = \{1,\ldots,|J|\}
\end{equation}
The set $J$ will be the set of all servers from which we obtain encoded symbols.  We will also make use of sets $J_u^a\subseteq J$ where $u \in [s]$ and $a \in [b]$, which are defined so that during the $u^{th}$ iteration we obtain the symbol $y^i_j(a)$ from every server $j\in J_u^a$.

For clarity of presentation, we will describe Steps 1.-2.\ in detail for the first iteration of the scheme, which will help elucidate the structure of the queries and responses of subsequent iterations.  \newline

1. \emph{Query Construction:} We select $mb$ codewords $d^{\ell,a} = [d^{\ell,a}(1)\cdots d^{\ell,a}(n)]$ uniformly at random from $D$ for $\ell \in [m]$ and $a\in[b]$.  For $\ell \in [m]$ and $j \in [n]$, define
\begin{equation}
d^\ell_j = \left[
d^{\ell,1}(j)  \cdots  d^{\ell,b}(j)
\right] \in \F_q^{1\times b}\quad \text{and}\quad
d_j = 
\left[d^1_j \cdots d^m_j\right]\in \F_q^{1\times mb}
\end{equation}
We partition $J_1 := [c]\subseteq J$ into $b$ subsets as follows:
\begin{equation}
J_1^1= \{1,\ldots,g\},\ J_1^2 = \{g+1,\ldots,2g\},\ \ldots\,,\ J_1^b = \{g(b-1),\ldots,gb = c\}
\end{equation} 
and define the queries $q^i_j$ by 
\begin{equation}\label{query_defn_1}
q^i_j = \left\{\begin{array}{ll}
d_j + e_{b(i-1)+a} & \text{if $j \in J_1^a$} \\
d_j & \text{if $j\not\in J_1$} \\
\end{array}\right.
\end{equation}
where $e_{b(i-1)+a}\in\F_q^{1\times mb}$ denotes the $(b(i-1) + a)^{th}$ standard basis vector. Thus for $j \in J_1^a$, the query $q^i_j$ is simply $d_j$ but with the entry $d^{i,a}(j)$ replaced with $d^{i,a}(j) + 1$.\newline

2. \emph{Responses:}
To understand the response vector $r^i$, we first calculate $r^i_j$ for $j\not\in J_1$.  We have
\begin{equation}
r^i_j = \langle q^i_j,y_j\rangle = \langle d_j,y_j\rangle = \sum_{\ell = 1}^m \langle d_j^\ell, y_j^\ell \rangle = \sum_{\ell = 1}^m \sum_{a = 1}^b d^{\ell,a}(j)y^\ell_j(a).
\end{equation}
For $j\in J_1^{a_0}$ for some $a_0\in [b]$, the same calculation reveals that
\begin{equation}
r^i_j = \sum_{\ell = 1}^m \sum_{a = 1}^b d^{\ell,a}(j)y^\ell_j(a) + y^i_j(a_0).
\end{equation}
We see that the value of the total response vector during the first iteration is
\begin{equation}
r^i 
=
\sum_{\ell = 1}^m\sum_{a = 1}^b\begin{bmatrix}
d^{\ell,a}(1)y_1^\ell(a) \\
\vdots \\
d^{\ell,a}(n)y_n^\ell(a)
\end{bmatrix}
+
\begin{bmatrix}
y^i_1(1) \\ 
\vdots \\
y^i_g(1) \\
\vdots \\
y^i_{g(b-1)}(b) \\
\vdots \\
y^i_c(b) \\
0_{(n-c)\times 1}
\end{bmatrix}
= \underbrace{\sum_{\ell = 1}^m \sum_{a = 1}^b d^{\ell,a}\star y^{\ell,a}}_{\in C\star D} + 
\begin{bmatrix}
y^i_1(1) \\ 
\vdots \\
y^i_g(1) \\
\vdots \\
y^i_{g(b-1)}(b) \\
\vdots \\
y^i_c(b) \\
0_{(n-c)\times 1}
\end{bmatrix}
\end{equation}
where $y^{\ell,a}\in C$ is the $a^{th}$ row of $y^\ell$. \newline

3. \emph{Iteration:} During the $u^{th}$ iteration for $u = 2,\ldots,s$, we repeat Steps 1.-2.\ but recursively define the subset $J_u^a\subseteq J$ to be the cyclic shift  of $J_{u-1}^a$ within $J$ to the right by $g$ indices.  Thus if $J_{u-1}^a = \{j_1,\ldots,j_g\}$, then
\begin{equation}
J_u^a = \{j_1+g,j_2+g,\ldots,j_g + g\}
\end{equation}
where if $j\in J_u^a$ satisfies $j>|J|$, it is replaced with $(j-1)$ (mod $|J|$) $+$ $1$.  We let $J_u = J_u^1\cup\cdots\cup J_u^b$ and define the queries during the $u^{th}$ iteration by
\begin{equation}\label{query_defn_2}
q^i_j = \left\{\begin{array}{ll}
d_{j} + e_{b(i-1)+a} & \text{if } j \in J_u^{a} \\
d_j & \text{if } j\not\in J_u 
\end{array}\right.
\end{equation}
The response vector $r^i$ during the $u^{th}$ iteration is of the form
\begin{equation}
r^i = (\text{codeword of $C\star D$}) + y^i_{J_u}
\end{equation}
where $y^i_{J_u}$ is a vector with entries $y_{j}^i(a)$ in some known positions for all $j\in J_u^a$ and all $a \in [b]$, and zeros elsewhere. \newline

4. \emph{Data Reconstruction:}  Let $S$ be a generator matrix for $(C\star D)^\perp$.  Since $c = d_{C\star D}-1$, every $c$ columns of $S$ are linearly independent.  To reconstruct the file $x^i$, we begin by considering the response vector $r^i$ from the first iteration, and computing
\begin{equation}
Sr^i
=
S(\text{codeword of $C\star D$})
+
S\begin{bmatrix}
y^i_1(1) \\ 
\vdots \\
y^i_c(b) \\
0_{(n-c)\times 1}
\end{bmatrix}
=
S\begin{bmatrix}
y^i_1(1) \\ 
\vdots \\
y^i_c(b) \\
0_{(n-c)\times 1}
\end{bmatrix}.
\end{equation}
From $Sr^i$ we can obtain the values of $y^i_1(1),\ldots,y^i_c(b)$, since the first $c$ columns of $S$ are linearly independent.  If $r^i$ is instead the response during the $u^{th}$ iteration of the scheme, we similarly obtain all entries of the form $y^i_{j}(a)$ for $j\in J_u^a$ and $a\in[b]$ from the product $Sr^i$.  For a fixed row $a$, the sets $J_1^a,\ldots,J_s^a$ are disjoint and consist of $sg = k$ servers in total, hence we retrieve $k$ distinct symbols of $y^{i,a}$ for every row $a$.\newline

One can visualize the entire PIR scheme as in Fig.\ \ref{scheme_fig}, wherein we show what portions of the encoded file $y^i$ we are downloading during each iteration, for parameters $k = 6$, $c = 4$, and $n = 10$ (in the left-hand figure).  Here each file consists of $b = 2$ rows, and the scheme requires $s = 3$ iterations.  We have $J = \{1,\ldots,k\}$ and the sets $J_u^a$ are given by
\begin{equation}
\begin{aligned}
J_1^1 &= \{1,2\},\quad J_1^2 = \{3,4\} \\
J_2^1 &= \{3,4\},\quad J_2^2 = \{5,6\} \\
J_3^1 &= \{5,6\},\quad J_3^2 = \{1,2\}
\end{aligned}
\end{equation}
In Fig.\ \ref{scheme_fig} we denote the encoded symbols downloaded during the first iteration in red, during the second iteration in blue, and during the third iteration in green.  

In the right-hand side of Fig. \ref{scheme_fig} we repeat this exercise for parameters $k = 4$, $c = 6$, and $n = 10$.  In this case we have $b = 3$ rows per file and require $s = 2$ iterations.  The sets $J_u^a$ are given by
\begin{equation}
\begin{aligned}
J_1^1 &= \{1,2\},\quad J_1^2 = \{3,4\},\quad J_1^3  = \{5,6\} \\
J_2^1 &= \{3,4\},\quad J_2^2 = \{5,6\},\quad J_2^3 = \{1,2\}
\end{aligned}
\end{equation}
which are depicted in Fig.\ \ref{scheme_fig}.

\begin{figure}
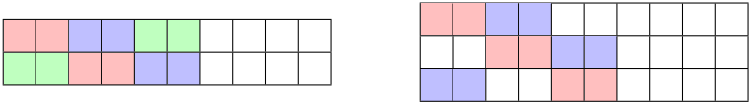

\centering
\begin{center}
\begin{tabular}{|c|c|c|c|c|c|c|c|c|c|}
\hline
\cellcolor{red!25} & \cellcolor{red!25} & \cellcolor{blue!25} & \cellcolor{blue!25} & \cellcolor{green!25} & \cellcolor{green!25} & & & & \\
\hline
\cellcolor{green!25} & \cellcolor{green!25} & \cellcolor{red!25} & \cellcolor{red!25} & \cellcolor{blue!25} & \cellcolor{blue!25} & & & & \\
\hline
\end{tabular}
\quad\quad\quad
\begin{tabular}{|c|c|c|c|c|c|c|c|c|c|}
\hline
\cellcolor{red!25} & \cellcolor{red!25} & \cellcolor{blue!25} & \cellcolor{blue!25} & & & & & & \\
\hline
& & \cellcolor{red!25} & \cellcolor{red!25} & \cellcolor{blue!25} & \cellcolor{blue!25} & & & & \\
\hline
\cellcolor{blue!25} & \cellcolor{blue!25} & & & \cellcolor{red!25} & \cellcolor{red!25}  & & & & \\
\hline
\end{tabular}
\end{center}
\caption{Visualizing the PIR scheme of Section \ref{construction}.  On the left, downloading from a system with parameters $(k,c,n) = (6,4,10)$.  Since $k>c$, we have $J = \{1,\ldots,k\}$, and the scheme requires $b = 2$ rows and $s = 3$ iterations.  On the right, a system with parameters $(k,c,n) = (4,6,10)$.  Here $c > k$, so $J = \{1,\ldots,c\}$ and the scheme requires $b = 3$ rows and $s = 2$ iterations.  Depicted is the encoded file $y^i\in \F_q^{b\times n}$, along with the encoded symbols downloaded in the first (red), second (blue), and third (green) iterations.  The columns which contain colored blocks are those in $J$.}\label{scheme_fig}
\end{figure}


\subsection{Proofs of Correctness and Privacy}\label{Proof}

In this section we provide proofs that the PIR scheme described in the previous subsection is correct (retrieves the desired file) and preserves privacy (does not reveal the index $i$ of the desired file to any group of $t$ colluding servers).

\begin{theorem}\label{correct}
Let $C$ be an $[n,k,d]$-code and suppose we have a retrieval code $D$ such that either (i) $d_{C\star D}-1 \leq k$, or
(ii) there exists $J\subseteq [n]$ of size $\max\{d_{C\star D}-1,k\}$ such that every subset of $J$ of size $k$ is an information set of $C$.  Then the PIR scheme of Section \ref{construction} is correct, that is, retrieves the desired file with rate $(d_{C\star D}-1)/n$.
\end{theorem}
\begin{proof}
If condition (i) is satisfied, we choose $J\subseteq [n]$ of size $k$ to be any information set of $C$.  In the data reconstruction phase of the PIR scheme, we retrieve $k$ symbols from each row $y^{i,a}$ of $y^i$, corresponding to the columns belonging to $J$.  Since every $K\subseteq J$ of size $k$ is an information set, this suffices to recover every $y^{i,a}$ and therefore all of $x^i$.  The rate of the scheme is easily seen to be
\begin{equation}
\frac{bk}{ns} = \frac{k\cdot \frac{\lcm(c,k)}{k}}{n\cdot \frac{\lcm(c,k)}{c}} = \frac{d_{C\star D}-1}{n}
\end{equation}
which completes the proof.
\end{proof}

\begin{theorem}\label{private}
Then the PIR scheme described in Section \ref{construction} protects against $d_{D^\perp}-1$ colluding servers.
\end{theorem}
\begin{proof}
Let $T = \{j_1,\ldots,j_t\} \subseteq [n]$ be a set of servers of size $t\leq d_{D^\perp}-1$.   We begin by showing that during a single iteration, we have $I(q^i_{j_1},\ldots,q^i_{j_t};i) = 0$.  From $t\leq d_{D^\perp}-1$ it follows immediately that every $t$ columns of the generator matrix of $D$ are linearly independent.  Therefore the code $D_T$ is the entire space $\F_q^t$.

First consider the distribution of one of the vectors
\begin{equation}
d_j = [d^{1,1}(j)\cdots d^{1,b}(j)\cdots d^{m,1}(j)\cdots d^{m,b}(j)]\in \F_q^{1\times bm}
\end{equation}
for a single $j\in T$.  As $D_{\{j\}}$ is distributed uniformly  on $\F_q$, and the codewords $d^{\ell,a}$ are selected uniformly at random from $D$, it follows that $d_j$ is uniform on $\F_q^{1\times bm}$.  

Similarly, as $D_T$ is all of $\F_q^t$, we see that the joint distribution $\{d_j\ |\ j\in T\}$ is uniform over $\left(\F_q^{1\times bm}\right)^t$.  If $f(i,j)$ denotes the index of the standard basis vector as in (\ref{query_defn_1}) and (\ref{query_defn_2}), we see that
\begin{equation}
\{q^i_{j_1},\ldots,q^i_{j_t}\} = \{d_j + e_{f(i,j)}\ |\ j\in T\cap J\} \cup \{d_j\ |\ j\in T\setminus J\}
\end{equation}
is uniformly distributed for all $i$, as translating the uniform distribution by any vector results again in the uniform distribution. The distribution $\{q^i_{j_1},\ldots,q^i_{j_t}\}$ of the queries is therefore independent of the index $i$ of the desired file, hence $I(q^i_{j_1},\ldots,q^i_{j_t};i) = 0$ is satisfied for a single iteration.

Now consider the joint distribution $Q^i_T$ of all queries to all servers in $T$, as we range over all iterations of the scheme.  For each iteration, the vectors $d^{\ell,a}$ are chosen independently of all other iterations, from which arguments identical to the above show that $Q^i_T$ is uniform on $\left(\F_q^{1\times bm}\right)^{ts}$.  Thus $I(Q^i_T;i) = 0$ as desired.
\end{proof}

\subsection{Examples}  In this subsection we show how some previously constructed PIR schemes fit into the general framework of our scheme.  Throughout this section, we assume that $i$ is the index of the file we wish to retrieve.  In case either $b=1$ or $s=1$, we will suppress all indices relating to rows or iterations, respectively.

\begin{example}\label{example1}
Let $C$ be any systematic $[n,k,d]$ storage code and set $D=\Rep(n)$ so that $C\star D = C$.  The above-outlined scheme has rate $(d_C -1)/n$ and as  
$d_{D^\perp} -1 = 1$, it only provides privacy against non-colluding servers ($t=1$).  For simplicity we assume $d_C-1| k$, and therefore require only $b = 1$ row per file but $s = k/(d_C-1)$ iterations.  We set $J = [k]$.

As $D=\Rep(n)$, sampling from $D$ uniformly at random amounts to sampling uniformly at random from $\F_q$ itself.  Thus the query construction in this example amounts to selecting a single vector $d_0 = [d^1\cdots d^m]\in \F_q^m$ uniformly at random, and setting $d_1 = \cdots = d_n = d_0$.  Now set $J_1 = [d_C-1]$ and define the queries by
\begin{equation}
q^i_j = \left\{ \begin{array}{ll}
d_0 + e_{i} & \text{if } j\in J_1 \\
d_0 & \text{if } j\not\in J_1
\end{array}
\right.
\end{equation}
The total response vector is then easily seen to be
\begin{equation}
r^i = \sum_{\ell = 1}^m d^\ell y^\ell + \begin{bmatrix}
y^i(1) \\ \vdots \\ y^i(d_C-1) \\ 0_{(n-(d_C-1))\times 1}
\end{bmatrix}
\end{equation}
Let $S$ be a generator matrix of $C^\perp$ which is in systematic form.  Then
\begin{equation}
Sr^i = \left[y^i(1) \cdots y^i(d_C-1)\right] = \left[x^i(1)\cdots x^i(d_C-1)\right]
\end{equation}
In the second iteration, we obtain $J_2$ by shifting $J_1$ to the next set of $d_C-1$ servers and repeat the above in the obvious way.  After $k/(d_C-1)$ iterations we recover $[y^i(1)\cdots y^i(k)]= x^i$.  Privacy against any single server is clear, since for a fixed $j$ the queries $q^i_j$ are independent samples of the uniform distribution on $\F_q^m$.

This is essentially a paraphrasing of the scheme of \cite[Theorem 1]{razan_salim} in the language of our construction. However, in \cite{razan_salim}, the scheme is only presented for MDS codes.
\end{example}

\begin{example}\label{example2}
Let $C$ be any $[n,k]$ MDS code and set $D = C^\perp$.  We have $(C\star D)^\perp =  \Rep(n)$ by Proposition \ref{star} (iv). Since $d_{C\star D} -1 = 1$ and $d_{D^\perp} -1 = n-k$, the above-outlined scheme has rate $1/n$ and provides privacy against any $n-k$ colluding servers.  

We have $b = 1$ row per file and require $s = k$ iterations of the scheme.  We set $J = [k]$ and $J_1 = \{1\}$, so that the queries in the first iteration are given by
\begin{equation}
q^i_j = \left\{\begin{array}{ll}
d_j + e_i & \text{if } j = 1 \\
d_j & \text{if } j > 1
\end{array}
\right.
\end{equation}
where the vectors $d_j\in \F_q^n$ are as in the construction.  The response vector is
\begin{equation}
r^i = \sum_{\ell = 1}^m d^\ell\star y^\ell + \begin{bmatrix}
y^i(1) \\ 0_{(n-1)\times 1}
\end{bmatrix}
\end{equation}
where $d^\ell\in D$ and $y^\ell\in C$.  As $(C\star D)^\perp = \Rep(n)$, the reconstruction function takes a particularly simple form.  In particular we can take $S = [1\cdots 1]$, and see immediately that $Sr = y^i(1)$.  Iterating this procedure $k$ times while setting $J_u = \{u\}$ for $u \in [k]$ yields $[y^i(1)\cdots y^i(k)]$, which suffice to reconstruct $x^i$ by the MDS property.

This is the scheme of \cite[Theorem 2]{razan_salim}, again rephrased in the context of our scheme.
\end{example}

\begin{example}\label{example3}
Let $C = \Rep(n)$ and let $D$ be any $[n,t]$ MDS code.  We have $d_{C\star D} -1= d_D -1= n-t$ and $d_{D^\perp} -1= t$.  The above-outlined scheme thus has rate $(n-t)/n$ and provides privacy against any $t$ colluding servers.   We have $b = n - t$ rows per file but require only $s = 1$ iteration of the scheme.  We set $J = [n-t]$.

With $d_j\in \F_q^{1\times m(n-t)}$ as in the scheme construction, the queries $q^i_j$ are of the form
\begin{equation}
q^i_j = \left\{\begin{array}{ll}
d_j + e_{(n-t)(i-1)+j} & \text{if $j\in J$}\\
d_j & \text{if $j\not\in J$}
\end{array}\right.
\end{equation}
which yields a response vector of the form
\begin{equation}
r^i = \sum_{\ell = 1}^m \sum_{a = 1}^{n-t} x^\ell(a)d^{\ell,a} + \begin{bmatrix}
x^i(1) \\
\vdots \\
x^i(n-t) \\
0_{t\times 1}
\end{bmatrix}
\end{equation}
If $S$ is a generator matrix of $D^\perp$ in systematic form, then $Sr^i = x^i$ which completes the retrieval scheme.
\end{example}



\subsection{Remarks on Scheme Construction}

Theorem \ref{correct} implies that for \emph{any} storage code $C$ at all, provided that we choose $D$ such that $d_{C\star D}-1\leq k$, we can choose $J$ to be any information set of $C$ and achieve rate $(d_{C\star D}-1)/n$ and protection against $d_{D^\perp}-1$ colluding servers.  In particular, if $C$ is any MDS code then we can choose any subset $J$ of servers of size $|J| = \max\{d_{C\star D}-1,k\}$.  Thus the scheme achieves the stated rate and protection against collusion for any MDS storage code $C$ and any retrieval code $D$.  

It is likely that condition (ii) in Theorem \ref{correct} is somewhat conservative.  We do not really need every subset of $J$ of size $k$ to be an information set of $C$, only subsets of the form $J_1^a\cup\cdots\cup J_s^a$ for $a\in [b]$, which index the servers retrieving symbols from the $a^{th}$ row of $y^i$.  However, we prefer to state condition (ii) as it is for the sake of simplicity.

Theorem \ref{private} implies that we must have $\supp(D) = [n]$ to achieve any non-trivial privacy with our scheme.  For if $\supp(D)\neq [n]$ then we would have some standard basis vector in $D^\perp$, implying $d_{D^\perp}-1 = 0$. This can be interpreted by saying that every server has to see some amount of randomness.

It may be the case that the user does not have the freedom to adjust the number of rows in a file.  For example, each file might be stored as a single row of $X$, in which case it is not obvious how to take advantage of a scheme which downloads more symbols per iteration than there are in a file.  One way to remedy this is to simply have the user download multiple files.  Thus one could rephrase Example \ref{example3} so that the user downloads $n-t$ files from the database, instead of one file which consists of $n-t$ symbols.  However, according to recent results \cite{bananaman2}, the capacity of such multi-message PIR is higher than that of single-message PIR.  Thus to make a more valid comparison with known rate and capacity results, we have chosen to describe our schemes as only retrieving one file.

While our interest in this paper is in download cost, we observe that the user in each iteration of our scheme uploads $bnm$ symbols from $\F_q$, for a total of $bnms$ total uploaded symbols. In particular, while the download cost does not depend on the number of files stored, the upload cost grows linearly in $m$.  While the upload cost is also depends linearly on $b$ and $s$, the size of the file does as well, so while the upload cost grows with these parameters, so does the total amount of privately downloaded information.

\section{Private Information Retrieval from GRS Codes}\label{PIRfromGRS}

In~\cite{razan_salim}, two PIR schemes were presented for arbitrary $[n,k]$ MDS codes, one of which had rate $\frac{1}{n}$ and protected against $t=n-k$ colluding servers, and the other of which had rate $\frac{n-k}{n}$ and $t=1$. These schemes are essentially variations of Example~\ref{example1} and Example~\ref{example2} in this paper. The authors asked if one can adapt their schemes in the ``intermediate regime'' where $1<t<n-k$ for arbitrary $n$ and $k$. In this section, we show how to do this via the construction in Section \ref{main_scheme} for some suitably chosen $[n,k]$ MDS storage codes, namely for GRS codes.

By Proposition \ref{star} we know that the class of all GRS codes associated to a fixed $n$-tuple $\alpha\in\F_q^n$ is closed under taking star products and duals. Moreover, while the dimension of a star product $C\star D$ of two generic codes can be as high as $\dim(C)\cdot\dim(D)$, in the case of GRS codes it is only $\dim(C)+\dim(D)-1$, which is useful when we want to maximize minimum distances. Indeed, the following theorem from \cite{randi13}, which can be seen as a multiplicative version of the Singleton bound, shows that among all storage codes, our PIR schemes gives the best privacy-rate tradeoff precisely for GRS storage codes.
\begin{theorem}[Immediate corollary of \cite{randi13}, Theorem 2 and \cite{mirandola_zemor}, Theorem 14]\label{MultSing}
Let $C_1$ and $C_2$ be linear codes  of dimension $k_1$ and $k_2$ and support $[n]$. Then
\begin{equation}
d_{C_1\star C_2}-1\leq\max\{0 , n - (k_1+k_2-1)\}.
\end{equation}
Conversely, if neither $C_1$, $C_2$, nor $(C_1\star C_2)^\perp$ is the length $n$ repetition code, then the above bound is an equality exactly when both $C_1$ and $C_2$ are GRS codes.
\end{theorem}
By  Proposition \ref{star} (iv), GRS codes satisfy Theorem~\ref{MultSing} with equality.
 The following theorem instantiates our scheme in the case where the storage code $C$ and retrieval code $D$ are GRS.

\begin{theorem}\label{theorem2}
Let $C=GRS_k(\alpha,v)$ and consider the distributed storage system $Y = XG_C$ as in Section \ref{DSS}. Then for all $t$ such that $1\leq t\leq n-k$, there exists a retrieval code $D$ such that the PIR scheme constructed in Section \ref{main_scheme} has rate $\frac{n-(k+t-1)}{n}$ and protects against any $t$ colluding servers.
\end{theorem}
\begin{proof}
We will give a linear code $D$ satisfying  $d_{D^\perp} - 1 = t$ and $d_{C\star D}-1 = n-(k+t-1)$; the theorem then follows immediately from Theorems \ref{correct} and \ref{private}.
Let $D=GRS_t(\alpha,u)$ for an arbitrary vector $u \in {\F_q^\times}^n$, since its dual is also MDS we see that $d_{D^\perp} = t+1$. Then $(C\star D)^\perp=GRS_{n-k-t+1}(\alpha,w)$ with $w$ given by \eqref{inverse_coeff}. It follows that $d_{C \star D}-1=n-(k+t-1)$ as desired.
\end{proof}

When $k+t > n$ our PIR scheme has rate zero, regardless of the storage code chosen. This can be seen readily from Theorem~\ref{MultSing}, as the retrieval code $D$ must have rank at least $t$, so $d_{C\star D}-1\leq\max\{0 , n-(k+t-1)\} = 0$ and thus the number of retrieved symbols per iteration is $d_{C\star D}-1= 0$.

When $k = 1$, that is, the data is stored via a replication system, our scheme provides a rate of $\frac{n-t}{n}$.  It is known by \cite{sun_jafar_2} that the capacity for PIR in this case is $\frac{1-t/n}{1-(t/n)^m}$.  Thus our scheme is asymptotically capacity-achieving, in that the resulting rates approach capacity as the number of files $m\rightarrow\infty$.  

Similarly, when $t = 1$, that is, without server collusion, our scheme provides a rate of $\frac{n-k}{n}$.  By \cite{bananaman} the capacity for PIR in this case is $\frac{1-k/n}{1-(k/n)^m}$, thus our schemes are again asymptotically capacity-achieving.  

The capacity of coded PIR with colluding servers is known when $m =2$ and $k = n-1$ by \cite{sun_jafar_MDS_TPIR}, but no general result is known for $k>1$ and $t>1$. A previous version of this article conjectured the following:
\begin{conjecture}[Disproven, see \cite{sun_jafar_MDS_TPIR}]~\label{capacity}
Let $C$ be an $[n,k,d]$ code that stores $m$ files via the distributed storage system $Y = XG_C$, and fix $1\leq t\leq n-k$. Any PIR scheme  for $Y$ that protects against any $t$ colluding servers  has rate at most $\frac{1-\frac{k+t-1}{n}}{1-{(\frac{k+t-1}{n})}^m}$.
\end{conjecture}
However, this conjecture was disproven in \cite{sun_jafar_MDS_TPIR}, where the authors exhibited an explicit PIR scheme for $m = 2$ files distributed over $n= 4$ servers using a rate $1/2$ storage code $C$, which protects against $t = 2$ collusion.  The exhibited scheme has rate $3/5$, while our conjectured capacity was $4/7$.

We will refrain from stating any further conjectures on the capacity of coded PIR with server collusion. However, the question remains open as to whether our schemes are asymptotically capacity-achieving as $m\rightarrow\infty$ for general $k$ and $t$.  This is consistent with the results in~\cite{sun_jafar_MDS_TPIR}, where it is also proven that, although positive retrieval rates are possible when $k+t>n$, the rates decrease to $0$ as $m\to \infty$. We further remark that the rate of our schemes do not depend on the number of files stored, and the field size required to achieve the rates of Theorem \ref{theorem2} is only $q \geq n$, needed to guarantee the existence of GRS codes.  This is in contrast with the capacity-achieving schemes of \cite{sun_jafar_1}, \cite{sun_jafar_2}, \cite{sun_jafar_MDS_TPIR}, wherein the field size grows as $q = O(n^m)$. Similarly, for the scheme of \cite{zhang_ge} for MDS coded data with colluding servers, which outperforms our scheme with the number of files $m$ is small, the field size is required to satisfy $q\geq O\left(\binom{n}{k}\right)$.

In Fig.\ \ref{PIR_rates_n12_m8} we plot for $n = 12$ servers the achievable PIR rates as a function of the number of colluding servers $t$, for various code rates $k/n$.  The black curve represents the capacity for the case of $k = 1$, that is, when the data is stored using a replication system \cite{sun_jafar_2}, while the asterisks represent the capacity obtained in \cite{bananaman} for the non-colluding case $t=1$ at different code rates. One can see that even for a relatively small number of files and relatively large amount of collusion, our scheme is quite close to capacity.

\begin{figure}
\centering
\includegraphics[width=.6\textwidth]{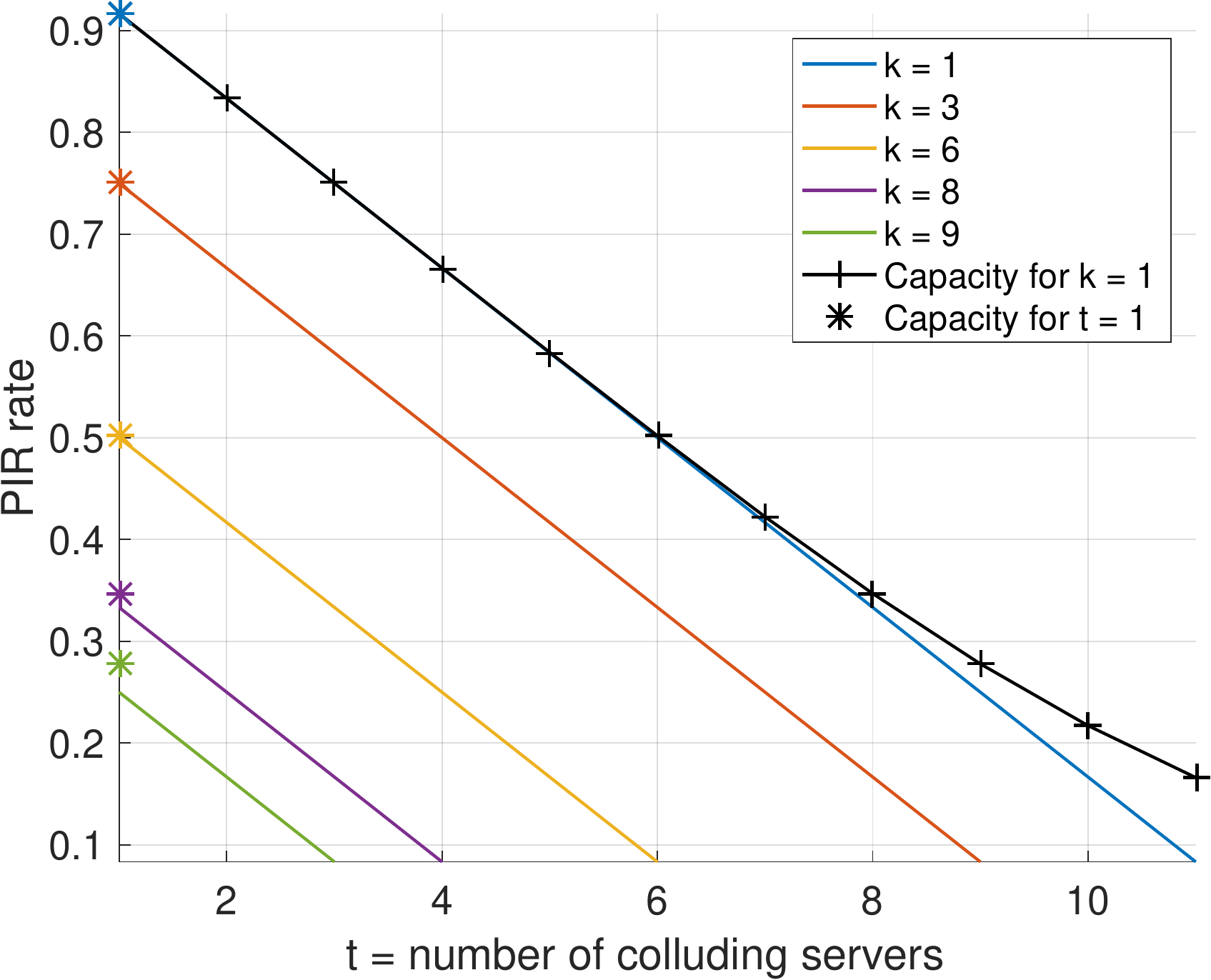}
\caption{Achievable PIR rates for $n = 12$ servers and $m = 8$ files, as a function of $t=$ the number of colluding servers, for various storage code rates. The black curve represents the PIR capacity for the case of $k=1$ (data stored via a replication system) as computed from \cite{sun_jafar_2}. The asterisks show the capacity for the non-colluding case $t=1$ as given in \cite{bananaman}. The PIR capacity is unknown when $t\geq 2$ and $k \geq 2$.}\label{PIR_rates_n12_m8}
\end{figure}

\section{An example in the intermediate regime}
 We will illustrate our scheme with an explicit example in the case when $t=2$ and $k = 2$. This is the first case not covered by our Examples~\ref{example1}--\ref{example3}. The storage code $C$ is MDS, with $[n,k,d]=[5,2,4]$.  We will have $c = k = 2$ and hence we require only $b = 1$ row per file and $s = 1$ iteration, thus we are free to ignore these parameters in what follows.

\begin{example}\label{bigexample}
Let $\alpha=[0,1,2,3,4]\in\F_5^5$ and let $\mathbf{1}\in\F_5^5$ be the all-ones vector. Consider the storage code $C=GRS_2(\alpha,\mathbf{1})$ over $\F_5$, encoded by its systematic generator matrix  \begin{equation}\label{eq:Cmatrix} G_C:=\begin{bmatrix}
		1 & 0 & 4 & 3 & 2\\
		0 & 1 & 2 & 3 & 4\\
	\end{bmatrix}
	 \end{equation}
	 as in~\eqref{eq:systematic}. So each file $x^i$ is divided into two blocks $x^i(1)$ and $x^i(2)$, and distributed onto the five servers as follows: 
 \begin{equation}\label{eq:Storage} \left\{\begin{array}{ll}
 x^i(1) & \mbox{on server } 1 \\
 x^i(2) & \mbox{on server } 2 \\
  4x^i(1)+2x^i(2) & \mbox{on server } 3 \\
   3x^i(1)+3x^i(2) & \mbox{on server } 4 \\
   2x^i(1)+4x^i(2) & \mbox{on server } 5.
\end{array}
	\right.
	 \end{equation}	 
	 
	 The random codewords $d^1,\ldots, d^m$ used to query the servers will be drawn from $D=GRS_2(\alpha,\mathbf{1})$, for which we choose the canonical generator matrix \begin{equation}\label{eq:Dmatrix} G_D:=\begin{bmatrix}
		1 & 1 & 1 & 1 & 1\\
		0 & 1 & 2 & 3 & 4
	\end{bmatrix}.
	 \end{equation} Note that $D^\perp$ is a $[5,3]$ MDS code, so our scheme protects against $t=d_{D^\perp}-1=2$ colluding servers. The reason we choose different generator matrices for $C$ and $D$ is practical, as the systematic generator matrix is better for decoding, while the canonical generator matrix is preferable for computations. 
	 
	 Observe that $C\star D=GRS_3(\alpha,\mathbf{1})$. We compute its dual $(C\star D)^\perp=GRS_2(\alpha,u)$, where $u_i = (\prod_{j \neq i}(\alpha_i-\alpha_j))^{-1}$ for $i=1,\ldots , 5$. Since $\alpha_i$ runs over the entire field $\F_5$, these products are unusually easy to evaluate, indeed we have $u_i=-1$ for all $i=1,\ldots , 5$. Thus $(C\star D)^\perp =GRS_2(\alpha, -\mathbf{1})=GRS_2(\alpha, \mathbf{1})$ so $(C\star D)^\perp $ is identical to $C$ and we use the same generator matrix $G_H=G_C$.

	 For each file index $\ell\in[m]$, we sample uniformly at random from $D$ by multiplying $G_D$ on the left by a uniform random vector $z^\ell = [z^\ell(1),z^\ell(2)]\in \F_5^2$, so that $d^\ell = z^\ell G_D$ and $d_j = [d^1(j)\cdots d^m(j)]$ for $j \in [m]$.  We let $z_1 = [z^\ell(1)\cdots z^m(1)]$ and $z_2 = [z^1(2)\cdots z^m(2)]$ which are independent and uniformly distributed over $\F_5^m$.
	 
	 Suppose we want to retrieve the file $x^i$ for some $i\in[m]$. We select $d_{C\star D}-1=2$ servers from which to download blocks from $x^i$, and for simplicity we here choose the systematic nodes. The queries $q^{i}_j$ sent to the servers will now be the following vectors in $\F_5^m$: 
\begin{equation}
\begin{aligned}
q^i_1 &= d_1 + e_i &=&\ z_1 + e_i \\
q^i_2 &= d_2 + e_i &=&\ z_1 + z_2 + e_i \\
q^i_3 &= d_3 &=&\ z_1 + 2z_2 \\
q^i_4 &= d_4 &=&\ z_1 + 3z_2 \\
q^i_5 &= d_5 &=&\ z_1 + 4z_2
\end{aligned}
\end{equation}
where $e_i$ is the $i^{th}$ standard basis vector.  Observe that for each pair of servers, the corresponding joint distribution of  queries is the uniform distribution over $(\F_5^m)^2$.
	 
	 The servers now respond by projecting their stored data onto the query vector, whence we obtain a response vector 
	 \begin{equation}\label{eq:Response} r^i = \begin{bmatrix}
\sum_{\ell = 1}^m d^\ell(1)x^\ell(1) &+ x^i(1)\\
\sum_{\ell = 1}^m (d^\ell(1)+d^\ell(2))x^\ell(2) &+ x^i(2) \\
  \sum_{\ell = 1}^m (d^\ell(1)+2d^\ell(2)) (4x^\ell(1)+2x^\ell(2))& \\
 \sum_{\ell = 1}^m (d^\ell(1)+3d^\ell(2))(3x^\ell(1)+2x^\ell(2)) &\\
  \sum_{\ell = 1}^m (d^\ell(1)+4d^\ell(2))(2x^\ell(1)+4x^\ell(2))&
\end{bmatrix}\in C\star D + \begin{bmatrix}
x^i(1)\\
x^i(2) \\
0 \\
0 \\
0
\end{bmatrix}.
	 \end{equation}	 
To finally decode the desired symbols, we now compute the matrix product $G_{(C\star D)^\perp}\cdot r^{i}$. One calculates that indeed 
\begin{equation}
G_{(C\star D)^\perp}\cdot r^{i}=G_{(C\star D)^\perp}\cdot\begin{bmatrix}
x^i(1)\\
x^i(2) \\
0 \\
0 \\
0
\end{bmatrix} = \begin{bmatrix}
x^i(1)\\
x^i(2) 
\end{bmatrix}.
\end{equation}
We have therefore extracted the two desired data blocks using $5$ queries, while maintaining privacy against $t=2$ colluding servers.
\end{example}

\bibliographystyle{siamplain}
\bibliography{references}

\end{document}